\documentclass[12pt]{article}

\usepackage{amsmath}
\usepackage{amssymb}
\usepackage{amsfonts}
\usepackage{latexsym}
\usepackage{bm}
\usepackage{color}

\catcode `\@=11 \@addtoreset{equation}{section}

\catcode `\@=12

\newcommand{\be}{\begin{equation}}
\newcommand{\en}{\end{equation}}
\newcommand{\bea}{\begin{eqnarray}}
\newcommand{\ena}{\end{eqnarray}}
\newcommand{\beano}{\begin{eqnarray*}}
\newcommand{\enano}{\end{eqnarray*}}
\newcommand{\bee}{\begin{enumerate}}
\newcommand{\ene}{\end{enumerate}}

\newcommand{\K}{{\mathfrak K}}

\newcommand{\mc}{\mathcal}

\newcommand{\D}{{\mc D}}

\newcommand{\F}{{\cal F}}

\newcommand{\Lc}{{\cal L}}

\newcommand{\1}{1 \!\! 1}

\newcommand{\Hil}{\mc H}

\newtheorem{thm}{Theorem}[section]

\newtheorem{prop}[thm]{Proposition}
\newtheorem{defn}[thm]{Definition}

\newenvironment{proof}{\noindent {\bf Proof --}}{\hfill$\square$ \vspace{3mm}\endtrivlist}

\catcode `\@=11 \@addtoreset{equation}{section}
\catcode `\@=12

\begin{document}

\thispagestyle{empty}

\vspace*{2cm}

\begin{center}
{{\Large \bf Intertwining operators for non self-adjoint Hamiltonians and bicoherent states\footnote{This paper is dedicated to the memory of Syed Twareque Ali,  much more than just a colleague!}}}\\[10mm]


{\large F. Bagarello} \footnote[1]{ Dipartimento di Energia, Ingegneria dell'Informazione e Modelli Matematici,
Facolt\`a di Ingegneria, Universit\`a di Palermo, I-90128  Palermo, and INFN, Universit\`a di di Torino, ITALY\\
e-mail: fabio.bagarello@unipa.it\,\,\,\, Home page: www.unipa.it/fabio.bagarello}

\end{center}

\vspace*{2cm}

\begin{abstract}
\noindent
This paper is devoted to the construction of what we will call {\em exactly solvable models}, i.e. of quantum mechanical systems described by an Hamiltonian $H$ whose eigenvalues and eigenvectors can be explicitly constructed out of some {\em minimal ingredients}. In particular, motivated by PT-quantum mechanics, we will not insist on any self-adjointness feature of the Hamiltonians considered in our construction. We also introduce the so-called bicoherent states, we analyze some of their properties and we show how they can be used for quantizing a system. Some examples, both in finite and in infinite-dimensional Hilbert spaces, are discussed.

\end{abstract}


\vfill


\newpage

\section{Introduction}

In recent years a growing interest on non self-adjoint operators with real eigenvalues has spread within the communities of physicists and of mathematicians, both for their possible applications to concrete (e.g., gain and loss) systems, and for the peculiar properties of these operators, which turn out to produce rather interesting mathematics in the Hilbert space where they are defined. A very recent book of collected papers (mainly) on the physical aspects of similar operators is \cite{bagspringer}, while many mathematical peculiarities are discussed in \cite{bagbook}. Concerning the latter, the role of biorthogonal sets, the spectral properties of these operators and the existence of similarity maps and of equivalent, or not equivalent, scalar products are just some of the points considered by several researchers in recent years.

Another very hot topic in physics has to do with the so-called intertwining operators, \cite{intop}. The reason for this is that, during the years, these operators have been used to construct more and more exactly solvable quantum mechanical models, i.e. Hamiltonians (usually, self-adjoint) for which the eigenvalues and the eigenvectors can be deduced in a reasonably simple way, out of a given {\em seed Hamiltonian} $H_1$, and of some  suitable operator $X$ which, in many concrete situations, is not invertible. $X$ is an intertwining operator for $H_1$ and $H_2$ if $H_2X=XH_1$. Hence, if $\varphi_n$ satisfies the eigenvalue equation $H_1\varphi_n=\epsilon_n\varphi_n$, then (at least if $\varphi_n\notin\ker(X)$) calling $\Psi_n=X\varphi_n$, this is an eigenstate of $H_2$: $H_2\Psi_n=\epsilon_n\Psi_n$.
 Of course, when $X$ is invertible, $H_2$ can be written as $H_2=XH_1X^{-1}$. In this case $H_1$ and $H_2$ are usually called similar\cite{footnote2}.  We will see examples of both these cases ($X^{-1}$ exists or not) in Section \ref{sect2}.

In this paper, we merge these two arguments, proposing a somehow general setting in which both these topics can be discussed and used to construct more models which are {\em under control}, i.e., are described by  certain Hamiltonian-like operators for which we can find eigenvalues and eigenvectors, even if we give up the assumption of self-adjointness of the Hamiltonian itself. Our analysis somehow extends what was first considered in \cite{bagintop}, where this kind of problems was first analyzed. Also, we discuss how these results can be used to construct some extended version of coherent states, and which are their properties. In particular, we focus our attention on the existence itself of these vectors, since this is not granted, in general. Then, we will deduce  the related resolutions of the identity, and we will discuss their nature of eigenstates of some particular lowering operators. Also, we will show their role in setting up a simple and interesting quantizing recipe.

This article is organized as follows:

In the next section we introduce the general framework for our intertwining operators for non self-adjoint operators, and we discuss the details of this construction, considering different situations depending on the different properties of the particular intertwining operator considered. Section 3 is devoted to some examples of this framework. In Section 4 we introduce our bicoherent states, and we show how these can be explicitly constructed by discussing a concrete example. The conclusions are given in Section 5.

\section{A general settings for non self-adjoint operators}\label{sect2}

Let us consider an operator $\Theta_1$ on a Hilbert space $\Hil$, in general different from its adjoint $\Theta_1^\dagger$, $\Theta_1\neq\Theta_1^\dagger$, with eigenvalues $\epsilon_n$ and eigenvectors $\varphi_n^{(1)}$:
\be
\Theta_1\varphi_n^{(1)}=\epsilon_n\varphi_n^{(1)},
\label{21}\en
for all $n\geq0$. In particular, in this section we will assume that all the eigenvalues have multiplicity one, but not that they are all real. This is relevant, for instance, in connection with PT-quantum mechanics in the so-called broken phase, \cite{benderrev}.  Quite often, in the literature, the assumption on the multiplicity of the eigenvalues is used just to simplify the notation, see \cite{fbjpa2016} for instance. Here, on the other hand, this is a more serious requirement, often satisfied in concrete models (see \cite{ben1}-\cite{dapro} among others), since it produces important consequences. This will be made clear later on.

For simplicity's sake, we will consider quite often the case in which $\dim(\Hil)<\infty$. This is because, in this case, we don't need to worry about the domains of the operators since they are all necessarily bounded. Nevertheless, quite often all along the paper we will also discuss what happens for infinite-dimensional Hilbert spaces, since this case is relevant in several physical applications. Roughly speaking, what we essentially want to discuss here is the possibility of constructing, out of $\Theta_1$ and of some other ingredient (see below), another (again, not necessarily self-adjoint) operator $\Theta_2$, whose eigenvalues and eigenstates can be deduced out of $\epsilon_n$ and $\varphi_n^{(1)}$. Moreover, since $\Theta_1\neq\Theta_1^\dagger$ and $\Theta_2\neq\Theta_2^\dagger$ in general, we are also interested in deducing, if this is possible, something on the eigenfamilies of $\Theta_1^\dagger$ and $\Theta_2^\dagger$. To be more explicit, the operator $\Theta_2$ we are looking for will be searched using some suitable (generalized) intertwining operator $X$, acting on $\Hil$, which could be invertible or not, bounded or not, and obeying (or not) further useful conditions, which will be clarified all along this section. In particular, in the rest of this section we will consider separately three different situations: we begin by assuming that $X^{-1}$ exists. Then we see what happens if $X^{-1}$ exists {\bf and} $XX^\dagger$ commutes with $\Theta_1$. Finally, we discuss what can be deduced when $X$ admits no inverse at all, but we still have $[XX^\dagger,\Theta_1]=0$. In this case, see Section \ref{sectcase3} below, we will need another assumption, i.e. the fact that $X^\dagger X$ has an inverse \cite{footnote3}.

\subsection{Possibility number 1: $X^{-1}$ exists}\label{sectcase1}

This is the easiest, and the most common situation, widely studied in the literature. In fact, in this case, it is clear that the operator $\Theta_2$ we are looking for can be naturally defined as
\be
\Theta_2=X^{-1}\Theta_1 X,
\label{22}\en
while its eigenvectors are $\tilde\varphi_n^{(2)}=X^{-1}\varphi_n^{(1)}$, for all $n$. The operator $\Theta_2$ is surely well defined on all of $\Hil$ if $\dim(\Hil)<\infty$. Otherwise, this is not granted, and we have to impose conditions on the domains. In this case, in fact, in order to have a well defined operator $\Theta_2$, the set of the vectors $f\in D(X)$, the domain of $X$, such that $Xf\in D(\Theta_1)$ and $\Theta_1 Xf\in D(X^{-1})$, must be dense in $\Hil$. The set of all these vectors is what we call $D(\Theta_2)$, the domain of $\Theta_2$.

Now, it is evident that $\Theta_2\tilde\varphi_n^{(2)}=\epsilon_n\tilde\varphi_n^{(2)}$, for all $n$, and that $X\Theta_2=\Theta_1X$. This case is not particularly interesting for us, since it only means that $\Theta_1$ and $\Theta_2$ are {similar} (i.e. related, as in (\ref{22}) by an invertible operator), and will not be considered further here. We just want to add that, if $\F_\varphi^{(1)}=\{\varphi_n^{(1)}\}$ is a basis for $\Hil$, then it exists an unique biorthogonal basis $\F_\psi^{(1)}=\{\psi_n^{(1)}\}$, \cite{chri}, $\left<\varphi_n^{(1)},\psi_k^{(1)}\right>=\delta_{n,k}$, and the vectors $\psi_n^{(1)}$ turn out to be eigenvectors of $\Theta_1^\dagger$, with eigenvalue $\overline{\epsilon_n}$:
$$
\Theta_1^\dagger\psi_n^{(1)}=\overline{\epsilon_n}\psi_n^{(1)},
$$
for all $n$. As for the eigenstates of $\Theta_2^\dagger$, these can also be easily found and turn out to be $\tilde\psi_n^{(2)}=X^{\dagger}\psi_n^{(1)}$. In fact, if $\psi_n^{(1)}\notin\ker(X^\dagger)$, a direct computation shows that $\Theta_2^\dagger\tilde\psi_n^{(2)}=\overline{\epsilon_n}\tilde\psi_n^{(2)}$. The role of $\ker(X^\dagger)$ is important here and will be stressed all along the paper, starting already from Section \ref{sectcase2}.

\subsection{Possibility number 2: $X^{-1}$ exists and $[XX^\dagger,\Theta_1]=0$}\label{sectcase2}

In this case, the situation is a bit more interesting and richer than before. Let us call $N_1=XX^\dagger$, and let us introduce $\Theta_2$ as in (\ref{22}), $\tilde\varphi_n^{(2)}=X^{-1}\varphi_n^{(1)}$, and  the new vectors $\varphi_n^{(2)}=X^\dagger\varphi_n^{(1)}$. Hence, as we will show now, something interesting can still be deduced, at least if $\varphi_n^{(1)}\notin \ker(X^\dagger)$. In this case, in fact, it is clear that, again, $\Theta_2\tilde\varphi_n^{(2)}=\epsilon_n\tilde\varphi_n^{(2)}$. Moreover,
$$
\Theta_2\varphi_n^{(2)}=\left(X^{-1}\Theta_1 X\right) X^\dagger\varphi_n^{(1)}=X^{-1}X X^\dagger\Theta_1\varphi_n^{(1)}=\epsilon_nX^\dagger\varphi_n^{(1)}= \epsilon_n \varphi_n^{(2)},
$$
where we have used the fact that $[N_1,\Theta_1]=0$.
This means that $ \varphi_n^{(2)}$ are also eigenstates of $\Theta_2$, with the same eigenvalues as  $\tilde\varphi_n^{(2)}$. But, recalling that the eigenvalues are assumed to have multiplicity one, this implies that $\tilde\varphi_n^{(2)}$ and $\varphi_n^{(2)}$ must be proportional:
\be
\tilde\varphi_n^{(2)}=k_n\varphi_n^{(2)}
\label{23}\en
for some non zero $k_n$. Of course, this proportionality between $\tilde\varphi_n^{(2)}$ and $\varphi_n^{(2)}$ could be lost when the multiplicity of some $\epsilon_n$ is larger than one.

\vspace{2mm}

{\bf Remarks:--} (1) It is useful to observe that $k_n=0$ if, and only if, $\varphi_n^{(1)}\in \ker(X^\dagger)$.

\vspace{1mm}

(2) Going back to the condition $[N_1,\Theta_1]=0$, at least if one  of these operators is unbounded, then it is convenient to assume (as it often happens in concrete examples) that a dense subset $\D$ of $\Hil$ exists which is stable under the action of $N_1$ and $\Theta_1$. If this is the case, then $[N_1,\Theta_1]=0$ must be understood in the following way: $N_1\Theta_1f=\Theta_1N_1f$, for all $f\in \D$. Several quantum mechanical systems having this feature are considered in \cite{baginbagbook}, in the context of the so-called $\D$ pseudo-bosons.

\vspace{2mm}

It is clear that $\Theta_1$ and $\Theta_2$ still obey the intertwining relation $X\Theta_2=\Theta_1X$. It is also clear that this equality should only hold on a dense subset of $\Hil$ if $\dim(\Hil)=\infty$ and if some of the operators involved are unbounded. An induction argument shows that $X\Theta_2=\Theta_1X$ can be extended to higher powers of $\Theta_1$ and $\Theta_2$: $X\Theta_2^n=\Theta_1^nX$, for all $n\geq1$. Even more: if $f(x)$ is a function which admits a power series expansion, with infinite convergence radius, then $Xf(\Theta_2)=f(\Theta_1)X$. Again, these two last intertwining equations should be considered on some suitable domain in case of unbounded operators. Otherwise, and in particular if $\dim(\Hil)<\infty$, they hold in all of $\Hil$.

Taking now the adjoint of $X\Theta_2=\Theta_1X$ we get another, equivalent, relation: $\Theta_2^\dagger X^\dagger=X^\dagger \Theta_1^\dagger$. However, under the conditions we are considering here ($[N_1,\Theta_1]=0$), it is also possible to deduce that $\Theta_2 X^\dagger=X^\dagger \Theta_1$ as well. In fact we have:
$$
\Theta_2 X^\dagger=X^{-1}\Theta_1 X X^\dagger= X^{-1} X X^\dagger \Theta_1=X^\dagger \Theta_1,
$$
which is what we had to check. Of course, the same arguments as above imply also that $X\Theta_2^\dagger= \Theta_1^\dagger X$,  $\Theta_2^n X^\dagger=X^\dagger \Theta_1^n$, for all $n\geq1$, and so on.

Let us now go back to the commutativity condition $[N_1,\Theta_1]=0$, and let us recall that $\varphi_n^{(1)}$ is an eigenstate of $\Theta_1$. It is not a big surprise to find that  $\varphi_n^{(1)}$ is also an eigenstate of $N_1$. In fact, equation (\ref{23}) can be rewritten as $X^{-1}\varphi_n^{(1)}=k_nX^\dagger\varphi_n^{(1)}$, for all $n$ such that $\varphi_n^{(1)}\notin\ker(X^\dagger)$.  Hence, left multiplying both sides of this equality for $X$, we get $\varphi_n^{(1)}=k_nN_1\varphi_n^{(1)}$ or, written in a more convenient way,
\be
N_1\varphi_n^{(1)}=k_n^{-1}\,\varphi_n^{(1)}.
\label{24}
\en
We observe that this equation is well defined since, as we have seen, $\varphi_n^{(1)}\notin\ker(X^\dagger)$ if and only if $k_n\neq0$. Incidentally, since $N_1$ is a positive operator, this implies that $k_n$ is strictly positive. The above implication can be inverted and, in fact,  equation (\ref{23}) can be easily deduced from (\ref{24}). Therefore, these two equations are equivalent, at least in absence of domain problems.
Needless to say,  the (possibly) most interesting situation here is when the eigenvalues $k_n^{-1}$ of $N_1$ are degenerate. In fact, if this is not so,  the various $\varphi_n^{(1)}$'s turn out to be also eigenstates of the (at least formally) self-adjoint operator $N_1$, and this would automatically imply that they are mutually orthogonal. When this happens, the two sets $\F_\varphi^{(1)}$ and $\F_\psi^{(1)}=\{\psi_n^{(1)}\}$ essentially coincide.

\subsection{Possibility number 3: $X^{-1}$ does not exist}\label{sectcase3}

This is probably the most interesting situation, since it is  quite different from what it is usually discussed in the literature. In fact, in this case, equation (\ref{22}) makes no sense, but  we can still see that, under suitable conditions on $X$, an operator $\Theta_2$ can still be defined, with the feature that its eigensystem can be easily deduced out of the one of $\Theta_1$. This will be again related to the existence of an intertwining relation between $\Theta_1$ and $\Theta_2$,  which can be deduced also in the present situation. What will be discussed here extends what was originally considered in \cite{bagintop}. The working assumptions are the following: the operator $N_1=XX^\dagger$ commutes with $\Theta_1$ as in Section \ref{sectcase2}, while the operator $N_2=X^\dagger X$ is strictly positive, hence invertible. As usual, we will mainly work under the assumption that $\dim(\Hil)<\infty$, to avoid any domain problem. However, also in view of the examples given in Section 3 and of the application to coherent states in Section 4, we will sometimes comment on the infinite dimensional case.

Under our assumptions, we introduce now
\be
\Theta_2=N_2^{-1}\left(X^\dagger \Theta_1 X\right),
\label{25}\en
while $\varphi_n^{(2)}$ are defined as in Section \ref{sectcase2},  $\varphi_n^{(2)}=X^\dagger\varphi_n^{(1)}$. Once again, the interesting situation is when  $\varphi_n^{(1)}\notin \ker(X^\dagger)$ to ensure that $\varphi_n^{(2)}\neq0$. But, as we will see in some explicit examples, this is not always granted. When this happens, the set of the eigenvalues of $\Theta_2$ turns out to be a proper subset of the set of eigenvalues of $\Theta_1$. We will meet explicitly with this situation in, e.g.,  Section \ref{sectex2}.

A first obvious remark is that, in this case, equation (\ref{23}) makes no sense, since the  vector $\tilde\varphi_n^{(2)}$ cannot even be defined. However, using our assumptions, we can still check that $\Theta_2\varphi_n^{(2)}= \epsilon_n \varphi_n^{(2)}$. The proof, which can be already found in \cite{bagintop}, goes like this
$$
\Theta_2\varphi_n^{(2)}=N_2^{-1}\left(X^\dagger \Theta_1 X\right)(X^\dagger \varphi_n^{(1)})=N_2^{-1}X^\dagger N_1\Theta_1\varphi_n^{(1)}=
N_2^{-1}N_2 X^\dagger (\epsilon_n\varphi_n^{(1)})=\epsilon_n\varphi_n^{(2)},
$$
and uses the fact that $\Theta_1N_1=N_1\Theta_1$.
Now it is also interesting to notice that, despite of the fact that definition (\ref{25}) is significantly different from (\ref{22}), $\Theta_2$ and $\Theta_1$ still satisfy the same intertwining relations as those found previously. To prove this, we first observe that $[N_2,X^\dagger \Theta_1 X]=0$:
$$
N_2 X^\dagger \Theta_1 X=  X^\dagger X(X^\dagger\Theta_1 X)= X^\dagger N_1\Theta_1 X=X^\dagger \Theta_1N_1 X=X^\dagger \Theta_1X X^\dagger X=X^\dagger \Theta_1 X N_2.
$$
Then we also have $[N_2^{-1},X^\dagger \Theta_1 X]=0$. Of course, this is true with no further assumption if $\dim(\Hil)<\infty$, while some more care is needed otherwise. Now we have
$$
X\Theta_2=X (X^\dagger \Theta_1 X)N_2^{-1}=N_1\Theta_1 XN_2^{-1}=\Theta_1N_1 XN_2^{-1}=\Theta_1 X N_2 N_2^{-1}=\Theta_1 X.
$$
As for the second intertwining relation, $X^\dagger \Theta_1=\Theta_2 X^\dagger$, we have
$$
\Theta_2 X^\dagger=N_2^{-1} (X^\dagger \Theta_1 X) X^\dagger =N_2^{-1} X^\dagger \Theta_1 N_1=N_2^{-1} X^\dagger N_1 \Theta_1 =N_2^{-1} N_2 X^\dagger \Theta_1 =X^\dagger \Theta_1,$$
as we had to prove. Moreover, it is clear that the following third intertwining relation is also satisfied: $XN_2=N_1X$, which shows that $X$ is also an intertwining operator between the operators $N_1$ and $N_2$, and not only between $\Theta_1$ and $\Theta_2$.

The following result replace, and correct, a similar statement in \cite{bagintop}:

\begin{prop}\label{prop1}
If $\Theta_1$, $N_1$ and $N_2$ are as above, and in particular if $[\Theta_1,N_1]=0$ and $N_2>0$, then: (i) $[N_2,\Theta_2]=0$; (ii) if $\Theta_1=\Theta_1^\dagger$,  then $\Theta_2=\Theta_2^\dagger$; (iii) if $\Theta_2=\Theta_2^\dagger$ and if $N_1>0$, then $\Theta_1=\Theta_1^\dagger$.
\end{prop}

\begin{proof}

(i) This claim follows from the definition (\ref{25}) of $\Theta_2$ and from the equality $[N_2,X^\dagger \Theta_1 X]=0$, which we have already proved.

(ii) Let $\Theta_1=\Theta_1^\dagger$. Then, since $N_2^\dagger=N_2$,
$$
\Theta_2^\dagger=\left(N_2^{-1}\left(X^\dagger \Theta_1 X\right)\right)^\dagger =\left(X^\dagger \Theta_1 X\right) N_2^{-1}=\Theta_2.
$$

(iii) Let now assume that $N_1>0$, so that $N_1^{-1}$ exists, and that $\Theta_2=\Theta_2^\dagger$. Then, left-multiplying this last equality by $X$ and right-multiplying it by $X^\dagger$, we get $X\Theta_2X^\dagger=X\Theta_2^\dagger X^\dagger$. Now, since $X\Theta_2=\Theta_1X$ and $X\Theta_2^\dagger=\Theta_1^\dagger X$, we get $X\Theta_2X^\dagger=\Theta_1 XX^\dagger=\Theta_1N_1$ and $X\Theta_2^\dagger X^\dagger = \Theta_1^\dagger XX^\dagger= \Theta_1^\dagger N_1$. Hence $\Theta_1N_1=\Theta_1^\dagger N_1$ which, since $N_1^{-1}$ exists, implies our statement.

\end{proof}

\vspace{2mm}

{\bf Remarks:--} (1) Similar results can be deduced also in the context of Section \ref{sectcase2}, but we will not repeat the proofs here.

\vspace{1mm}

(2) In \cite{bagintop} it was stated that $\Theta_1=\Theta_1^\dagger$ if and only if $\Theta_2=\Theta_2^\dagger$, with no requirement on the invertibility of $N_1$. This was not correct. We notice that the statement here is more natural than the one in \cite{bagintop}, since the role of $N_1$ and $N_2$ is now completely symmetric.

\vspace{1mm}

(3) Under the same assumptions of Proposition \ref{prop1} it is possible to check that $XN_2^{-1}=N_1^{-1}X$ and that $\Theta_1=N_1^{-1}(X\Theta_2X^\dagger)$. This formula is completely analogous to the one in (\ref{25}).

\subsection{The eigensystems for $\Theta_1^\dagger$ and $\Theta_2^\dagger$}

Starting now with the set $\F_\varphi^{(1)}=\{\varphi_n^{(1)}\}$ of eigenstates of $\Theta_1$, see (\ref{21}), and assuming that this set is a basis for $\Hil$, it is possible to define uniquely, \cite{chri}, a second set, $\F_\psi^{(1)}=\{\psi_n^{(1)}\}$, which is still a basis for $\Hil$ and  is biorthogonal to $\F_\varphi^{(1)}$: $\left<\varphi_k^{(1)},\psi_n^{(1)}\right>=\delta_{k,n}$. It is clear that both these sets are complete: if $f\in\Hil$ is orthogonal to all the $\varphi_n^{(1)}$'s (or to all the $\psi_n^{(1)}$'s), $f$ is necessarily zero. Then it is easy to check that, not surprisingly, the vectors in $\F_\psi^{(1)}$ are eigenstates of $\Theta_1^\dagger$:
\be
\Theta_1^\dagger\psi_n^{(1)}=\overline{\epsilon_n}\psi_n^{(1)},
\label{26}\en
for all $n$. In fact we have
$$
\left<\left(\Theta_1^\dagger\psi_n^{(1)}-\overline{\epsilon_n}\psi_n^{(1)}\right),\varphi_n^{(1)}\right>=\left<\psi_n^{(1)},\Theta_1
\varphi_k^{(1)}\right>-\epsilon_n\left<\psi_n^{(1)},\varphi_k^{(1)}\right>=0
$$
for all $n$ and $k$. Formula (\ref{26}) now follows from the completeness of $\F_\varphi^{(1)}$.

Now we want to construct, out of $\Theta_1^\dagger$, a second operator, which we could call $(\Theta_1^\dagger)_2$, mimicking what we have done in (\ref{25}). This is  possible since $[N_1,\Theta_1]=0$ implies that  $[N_1,\Theta_1^\dagger]=0$ as well. Hence we are under the working assumptions listed at the beginning of this section, with $\Theta_1$ replaced by $\Theta_1^\dagger$. Therefore  we can define
\be
(\Theta_1^\dagger)_2=N_2^{-1}\left(X^\dagger \Theta_1^\dagger X\right).
\label{add1}\en
It is interesting to notice that, if we take the adjoint of $\Theta_2$ in (\ref{25}) and we use the fact that $[N_2^{-1},X^\dagger\Theta_1X]=0$, we get $(\Theta_1^\dagger)_2=\Theta_2^\dagger$. In other words: taking the adjoint of the operator $\Theta_2$ in (\ref{25}) or defining, again as in (\ref{25}), the new operator in (\ref{add1}), makes no difference.

Let us now go back to the existence of the non zero vector $\varphi_n^{(2)}$. As we have already discussed, $\varphi_n^{(2)}\neq0$ if and only if $\varphi_n^{(1)}\notin\ker(X^\dagger)$. In \cite{bagintop} we have shown that $\varphi_n^{(1)}\in\ker(X^\dagger)$ if and only if $\varphi_n^{(2)}\in\ker(X)$. It is now easy to check that, if $\varphi_n^{(1)}\notin\ker(X^\dagger)$, then $\varphi_n^{(1)}\notin\ker(N_1)$ and $\varphi_n^{(2)}\notin\ker(N_2)$. In the examples and in Section \ref{sectCS} this aspect will be considered further.

Now, following what we did for $\F_\varphi^{(1)}$ and $\F_\varphi^{(2)}$, we define new vectors $\psi_n^{(2)}=X^\dagger\psi_n^{(1)}$, and the related set $\F_\psi^{(2)}=\{\psi_n^{(2)}\}$. Of course, $\psi_n^{(2)}\neq0$ if $\psi_n^{(1)}\notin\ker(X^\dagger)$.

A direct computation, based on the fact that $[N_1,\Theta_1^\dagger]=0$, shows that each $\psi_n^{(2)}\neq0$ is an eigenstate of $\Theta_2^\dagger$:
$$
\Theta_2^\dagger\psi_n^{(2)}=N_2^{-1}(X^\dagger\Theta_1^\dagger X)X^\dagger\psi_n^{(1)}=N_2^{-1}X^\dagger N_1\Theta_1^\dagger \psi_n^{(1)}=X^\dagger\overline{\epsilon_n}\psi_n^{(1)}=\overline{\epsilon_n}\psi_n^{(2)}.
$$

\vspace{2mm}

We will now prove that an eigenvalue equation like the one in (\ref{24}), deduced under the assumption that $X^{-1}$ does exist, can also be found now, and can even be extended to $N_2$. What we only need is that the eigenvalues $\epsilon_n$ of $\Theta_1$ have multiplicity one, and that $\varphi_n^{(1)}\notin\ker(X^\dagger)$. In fact, under this last condition, $N_1\varphi_n^{(1)}\neq0$. Moreover, since $\Theta_1(N_1\varphi_n^{(1)}) = N_1(\Theta_1\varphi_n^{(1)})=\epsilon_n(N_1\varphi_n^{(1)})$, the vector $N_1\varphi_n^{(1)}$ must be proportional to $\varphi_n^{(1)}$ itself. To distinguish here the situation with respect to that considered in Section \ref{sectcase2}, we call $\tilde k_n$ (instead of $k_n^{-1}$) this proportionality constant. Hence we get
\be
N_1\varphi_n^{(1)}=\tilde k_n\varphi_n^{(1)}.
\label{27}\en
Now, left-multiplying both sides of this equation with $X^\dagger$ we get $(X^\dagger X)(X^\dagger \varphi_n^{(1)})=\tilde k_n X^\dagger \varphi_n^{(1)}$, which can be rewritten as
\be
N_2\varphi_n^{(2)}=\tilde k_n\varphi_n^{(2)}.
\label{28}\en
Hence $N_1$ and $N_2$ have the same eigenvalue $\tilde k_n$, at least for those $n$'s for which $\varphi_n^{(1)}\notin\ker(X^\dagger)$. Notice also that, for all these $n$, we get
$$
\tilde k_n =\left(\frac{\|\varphi_n^{(2)}\|}{\|\varphi_n^{(1)}\|}\right)^2>0.
$$
As already observed previously, it is clear that, if the multiplicity of the eigenvalues $\tilde k_n$ of $N_1$ and $N_2$ is one, this makes of $\F_\varphi^{(1)}$ and $\F_\varphi^{(2)}$ two different sets each made of mutually orthogonal vectors. Hence, if we are interested to  non-o.n. sets of eigenvectors of $\Theta_1$ and $\Theta_2$, it is crucial that not all the $\tilde k_n$ have multiplicity one.

Equation (\ref{27}) can also be written as $XX^\dagger\varphi_n^{(1)}=\tilde k_n\varphi_n^{(1)}$, i.e. as $X\varphi_n^{(2)}=\tilde k_n\varphi_n^{(1)}$, which implies that $\varphi_n^{(1)}=\frac{1}{\tilde k_n}\,X \varphi_n^{(2)}$. This equality can be considered as the {\em inverse} of the equation $\varphi_n^{(2)}=X^\dagger \varphi_n^{(1)}$ in our context.

\vspace{2mm}

Summarizing the main outcome of our analysis we can say that {\em out of the triple $(\Theta_1,\{\epsilon_n\},\F_\varphi^{(1)})$ we can produce three more triples which are associated, in the way discussed before, to other exactly solvable models: $(\Theta_2,\{\epsilon_n\},\F_\varphi^{(2)})$, $(\Theta_1^\dagger,\{\overline\epsilon_n\},\F_\Psi^{(1)})$ and $(\Theta_2^\dagger,\{\overline\epsilon_n\},\F_\Psi^{(2)})$.}

\vspace{2mm}

It should be noticed that an important difference between the pairs $(\F_\varphi^{(1)},\F_\psi^{(1)})$ and $(\F_\varphi^{(2)},\F_\psi^{(2)})$ exists, and  is the following: while the first pair satisfies the orthonormality condition $\left<\varphi_k^{(1)},\psi_n^{(1)}\right>=\delta_{k,n}$ by construction, the second satisfies the slightly different condition: $\left<\varphi_k^{(2)},\psi_n^{(2)}\right>=\tilde k_n\delta_{k,n}$. In fact, we have
\be
\left<\varphi_k^{(2)},\psi_n^{(2)}\right>=\left<X^\dagger\varphi_k^{(1)},X^\dagger\psi_n^{(1)}\right>= \left<N_1\varphi_k^{(1)},\psi_n^{(1)}\right>=\tilde k_n\left<\varphi_k^{(1)},\psi_n^{(1)}\right>= \tilde k_n\delta_{k,n}.
\label{28bis}\en
Of course, strict biorthonormality could be recovered by changing a little bit the definition of, say $\psi_n^{(2)}$,  by replacing the original formula, $\psi_n^{(2)}=X^\dagger\psi_n^{(1)}$, with $\psi_n^{(2)}=\frac{1}{\tilde k_n}\, X^\dagger\psi_n^{(1)}$. However, this would break down the original symmetry between $(\F_\varphi^{(1)},\F_\psi^{(1)})$ and $(\F_\varphi^{(2)},\F_\psi^{(2)})$, and we prefer not to adopt this alternative definition here.

\vspace{3mm}

It is now interesting to stress that, in all the situations considered above, i.e. both when $X^{-1}$ exists and when it doesn't, we have been able to deduce intertwining relations between $\Theta_1$ and $\Theta_2$. In particular, independently of the particular definition adopted here for $\Theta_2$, (\ref{22}) or (\ref{25}), we have always deduced that $X\Theta_2=\Theta_1X$. The Proposition below shows that this equation is really a key feature of our formulation, and not just a consequence of some {\em smart} definitions.

\begin{prop}\label{prop22}
Suppose two operators $\Theta_1$ and $\Theta_2$ satisfy the equation $X\Theta_2=\Theta_1X$, for some suitable $X$. Hence: (i) if $X$ is invertible, then $\Theta_2=X^{-1}\Theta_1X$; (ii) if $X$ is not invertible, but $N_2=X^\dagger X$ is invertible, then $\Theta_2=N_2^{-1}(X^\dagger\Theta_1X)$. In this case, taken an eigenvector $\varphi_n^{(1)}$ of $\Theta_1$,  if we further have $[XX^\dagger,\Theta_1]=0$ and if $\varphi_n^{(1)}\notin\ker(X^\dagger)$, then $\varphi_n^{(2)}=X^\dagger\varphi_n^{(1)}$ is an eigenvector of $\Theta_2$, with the same eigenvalue as $\varphi_n^{(1)}$.
\end{prop}

\begin{proof}
The proof of (i) is trivial. As for (ii), we observe that left-multiplying  $X\Theta_2=\Theta_1X$ for $X^\dagger$, and using the fact that $N_2=X^\dagger X$ is invertible, we deduce that $\Theta_2=N_2^{-1}(X^\dagger\Theta_1X)$. The rest of the statement is clear.

\end{proof}

\vspace{3mm}

{\bf Remarks:--} (1) In \cite{fbjpa2016} we have discussed the role of antilinear operators for some intertwining relations between operators with complex eigenvalues. At a first sight,  this might appear not extremely far away from what we have done here. However, the two situations are indeed very different since the assumptions on $N_1$ and $N_2$ considered here play a crucial role in what is done in this section, while similar conditions are completely absent in \cite{fbjpa2016}.

\vspace{1mm}

(2) It may be interesting to notice that the above general scheme can be easily extended to a slightly more general situation, i.e. to the case in which $\Theta_1$ and $\Theta_2$ act on two different Hilbert spaces, $\Hil_1$ and $\Hil_2$. In this case, of course, $X$, $X^{-1}$ (if it exists) and $X^\dagger$ are operators between $\Hil_1$ and $\Hil_2$ or viceversa, $N_1$ and $N_2$ also act on different Hilbert spaces, and $\varphi_n^{(1)}$ and $\varphi_n^{(2)}$ belong to these different spaces. However, with some minor (and obvious) modifications, all the results proved in this section can again be deduced. We will see a concrete example of this situation in Section \ref{sectex2}.

\section{Examples}

In this section we give some examples of our general settings, considering first two finite and then two infinite-dimensional cases. Among other aspects, we will see that it is not so rare that some $\varphi_n^{(1)}$ belongs to $\ker(X^\dagger)$, and we will see how this implies that the set of eigenvalues of $\Theta_2$ is just a proper subset of the set of the eigenvalues of $\Theta_1$, as already mentioned in Section 2.

\subsection{A first finite-dimensional example}

A first no-go result suggests that, if $\dim(\Hil)<\infty$, there is no square matrix $X$ such that $X$ is not invertible and still $N_2$ is strictly positive. This is because, since $\det(X)=0$, $\det(N_2)=\det(X^\dagger X)=0$. Then,  $N_2^{-1}$ does not exist and the working conditions of Section \ref{sectcase3} are never satisfied. However, we are left with the possibility of using the approaches discussed in Sections \ref{sectcase1} or \ref{sectcase2}.

For concreteness, we consider now the matrix
$$
X=\left(
    \begin{array}{cc}
      x_{11} & x_{12} \\
      -\overline{x_{12}} & \overline{x_{11}} \\
    \end{array}
  \right)
$$
With this particularly simple choice, we have $N_1=N_2=\tilde x\1$, where $\1$ is the $2\times2$ identity matrix and $\tilde x=|x_{11}|^2+|x_{12}|^2=\det(X)$. It is clear that, if $x_{11}$ or $x_{12}$ are non zero, $N_2^{-1}$ does exist. Also, $[N_1,\Theta_1]=0$ for any two-by-two matrix $\Theta_1$. However, since $\det(X)=\tilde x >0$, it is clear that $X^{-1}$  exists as well. Hence we could adopt the strategies proposed in Sections  \ref{sectcase1} and  \ref{sectcase2}. We will not make any explicit choice of $\Theta_1$ here. Rather than this, we observe that $X^{-1}=(\tilde x)^{-1}\,X^\dagger$. This implies that $\tilde\varphi_n^{(2)}=X^{-1}\varphi_n^{(1)}=(\tilde x)^{-1}\,X^\dagger\varphi_n^{(1)}=(\tilde x)^{-1}\varphi_n^{(2)}$, which is in complete agreement with (\ref{23}).

\subsection{A second finite-dimensional example}\label{sectex2}

A possible way out from the previous no-go result consists in considering an operator $X$ between different Hilbert spaces, $\Hil_1$ and $\Hil_2$. In this case, if $\dim(\Hil_1)\neq\dim(\Hil_2)$, with $\dim(\Hil_j)<\infty$, $j=1,2$, we can still construct examples of the framework discussed in Section \ref{sectcase3}. Let us consider the following non self-adjoint operator $\Theta_1$, defined on $\Hil_1={\Bbb C}^3$:

{\scriptsize
$$
\Theta_1=\left(
           \begin{array}{ccc}
             \frac{1}{6}((5+\sqrt{3})E_1-(1+\sqrt{3})E_2+2E_3) & \frac{1}{6}((1+\sqrt{3})E_1-(2+\sqrt{3})E_2+E_3) & \frac{1}{6}((-7-3\sqrt{3})E_1+(5+3\sqrt{3})E_2+2E_3) \\
             \frac{1}{3}((1-2\sqrt{3})E_1+2(-1+\sqrt{3})E_2+E_3) & \frac{1}{3}(-2E_1+4E_2+E_3) & \frac{1}{3}((1+2\sqrt{3})E_1-2(1-\sqrt{3})E_2+E_3) \\
             \frac{1}{6}((-7+3\sqrt{3})E_1+(5-3\sqrt{3})E_2+2E_3) & \frac{1}{6}((1-\sqrt{3})E_1+(-2+\sqrt{3})E_2+E_3) & \frac{1}{6}((5-\sqrt{3})E_1+(-1+\sqrt{3})E_2+2E_3) \\
           \end{array}
         \right)
$$
}
with eigenvectors
$$
\varphi_1^{(1)}=\left(
                  \begin{array}{c}
                    -\frac{1}{\sqrt{2}}-\frac{1}{\sqrt{6}} \\
                    \sqrt{\frac{2}{3}} \\
                    \frac{1}{\sqrt{2}}-\frac{1}{\sqrt{6}} \\
                  \end{array}
                \right), \quad \varphi_2^{(1)}=\left(
                  \begin{array}{c}
                    -\sqrt{\frac{2}{3}}-\frac{1}{\sqrt{2}} \\
                    2\sqrt{\frac{2}{3}} \\
                    -\sqrt{\frac{2}{3}}+\frac{1}{\sqrt{2}} \\
                  \end{array}
                \right), \quad \varphi_3^{(1)}= \frac{1}{\sqrt{3}}\left(
                  \begin{array}{c}
                    1 \\
                    1 \\
                    1 \\
                  \end{array}
                \right).
$$
We have $\Theta_1\varphi_n^{(1)}=E_n\varphi_n^{(1)}$, $n=1,2,3$. Here we assume that $E_j\in {\Bbb R}$, for all $j$. We further define
$$
X=\left(
    \begin{array}{cc}
      0 & 1 \\
      -\frac{\sqrt{3}}{2} & -\frac{1}{2} \\
      \frac{\sqrt{3}}{2} & -\frac{1}{2} \\
    \end{array}
  \right), \quad \Rightarrow N_1=XX^\dagger=\left(
                                              \begin{array}{ccc}
                                                1 & -\frac{1}{2} & -\frac{1}{2} \\
                                                -\frac{1}{2} & 1 & -\frac{1}{2} \\
                                                -\frac{1}{2} & -\frac{1}{2} & 1 \\
                                              \end{array}
                                            \right), \quad N_2=X^\dagger X= \frac{3}{2}\1_2,
$$
where $\1_2$ is the identity matrix in $\Hil_2={\Bbb C}^2$. Of course $N_2^{-1}$ exists while $X^{-1}$ does not\cite{footnote1}. Moreover, $\Theta_1 N_1=N_1\Theta_1$. Hence, we are in the conditions of Section \ref{sectcase3}. Now:
$$
\varphi_1^{(2)}=X^\dagger\varphi_1^{(1)}=\left(
                                           \begin{array}{c}
                                             \frac{-3+\sqrt{3}}{2\sqrt{2}} \\
                                             -\frac{1}{2}\sqrt{3(2+\sqrt{3})} \\
                                           \end{array}
                                         \right), \,
                                         \varphi_2^{(2)}=X^\dagger\varphi_2^{(1)}=\left(
                                           \begin{array}{c}
                                             \frac{-6+\sqrt{3}}{2\sqrt{2}} \\
                                             -\frac{3+2\sqrt{3}}{2\sqrt{2}} \\
                                           \end{array}
                                         \right), \, \varphi_3^{(2)}=X^\dagger\varphi_3^{(1)}=\left(
                                           \begin{array}{c}
                                             0 \\
                                             0 \\
                                           \end{array}
                                         \right),
$$
which shows, in particular, that $\varphi_3^{(1)}\in \ker(X^\dagger)$. This is not surprising, since $\dim(\Hil_2)=2$, and for this reason we can only have, at most, two linearly independent vectors in $\Hil_2$. We also find, using (\ref{25}),
$$
\Theta_2=\frac{1}{4}\left(
                      \begin{array}{cc}
                        -(1+\sqrt{3})E_1+(5+\sqrt{3})E_2 & (7-3\sqrt{3})(E_1-E_2) \\
                        -(5+3\sqrt{3})(E_1-E_2) & (5+\sqrt{3})E_1-(1+\sqrt{3})E_2 \\
                      \end{array}
                    \right),
$$
and we see that $E_3$ does not appear anymore in $\Theta_2$. This is in agreement with the fact that $\Theta_2\varphi_n^{(2)}=E_n\varphi_n^{(2)}$, $n=1,2$. Moreover, an explicit computation shows that $X\Theta_2=\Theta_1X$ and that $X^\dagger\Theta_1=\Theta_2 X^\dagger$.

Now the vectors of $\F_\psi^{(1)}$ can be found to be
$$
\psi_1^{(1)}=\left(
                  \begin{array}{c}
                    -\sqrt{2}+\frac{1}{\sqrt{6}} \\
                    -\sqrt{\frac{2}{3}} \\
                    \sqrt{2}+\frac{1}{\sqrt{6}} \\
                  \end{array}
                \right), \quad \psi_2^{(1)}=\left(
                  \begin{array}{c}
                    \frac{1}{\sqrt{2}}-\frac{1}{\sqrt{6}} \\
                    \sqrt{\frac{2}{3}} \\
                    -\sqrt{\frac{2}{3}+\frac{1}{\sqrt{3}}} \\
                  \end{array}
                \right), \quad \psi_3^{(1)}= \frac{1}{\sqrt{3}}\left(
                  \begin{array}{c}
                    1 \\
                    1 \\
                    1 \\
                  \end{array}
                \right).
$$
They are biorthogonal to  $\F_\varphi^{(1)}$, $\left<\varphi_n^{(1)},\psi_m^{(1)}\right>=\delta_{n,m}$, and, as expected, they are eigenstates of $\Theta_1^\dagger$: $\Theta_1^\dagger\psi_n^{(1)}=E_n\psi_n^{(1)}$, $n=1,2,3$. Now
$$
\psi_1^{(2)}=X^\dagger\psi_1^{(1)}=\left(
                                           \begin{array}{c}
                                             \sqrt{\frac{3}{2}}+\frac{3}{2\sqrt{2}} \\
                                             \frac{-6+\sqrt{3}}{2\sqrt{2}} \\
                                           \end{array}
                                         \right), \qquad
                                         \psi_2^{(2)}=X^\dagger\psi_2^{(1)}=\left(
                                           \begin{array}{c}
                                             -\frac{1}{2}\sqrt{3(2+\sqrt{3})} \\
                                             \frac{1}{2}\sqrt{3(2-\sqrt{3})} \\
                                           \end{array}
                                         \right),
$$
while we find that
$$
\psi_3^{(2)}=X^\dagger\psi_3^{(1)}=\left(
                                           \begin{array}{c}
                                             0 \\
                                             0 \\
                                           \end{array}
                                         \right).
                                         $$
Hence $\psi_3^{(1)}\in\ker(X^\dagger)$, and $\Theta_2^\dagger\psi_n^{(2)}=E_n\psi_n^{(2)}$, $n=1,2$. Moreover we get  $\left<\varphi_n^{(2)},\psi_m^{(2)}\right>=\frac{3}{2}\delta_{n,m}$, $n=1,2$. Therefore $\tilde k_1=\tilde k_2=\frac{3}{2}$.

\vspace{2mm}

{\bf Remark:--} It might be interesting to observe that the matrix $\Theta_2$ can be written in terms of pseudo-fermionic operators, see \cite{baginbagbook}. Two such operators $a$ and $b$, satisfy $\{a,b\}=\1$, $a^2=b^2=0$, and can be represented for instance as
$$
a=\alpha_{12}\left(
               \begin{array}{cc}
                 \alpha & 1 \\
                 -\alpha^2 & -\alpha \\
               \end{array}
             \right), \qquad b=\beta_{12}\left(
               \begin{array}{cc}
                 \beta & 1 \\
                 -\beta^2 & -\beta \\
               \end{array}
             \right),
$$
where, calling $\gamma^2=-\alpha_{12}\beta_{12}$, the parameters must be such that $(\alpha-\beta)^2\gamma^2=1$. In \cite{baggarg} we have seen that the general Hamiltonian $H=\omega ba+\rho\1$ takes the form
\be
H=\left(
               \begin{array}{cc}
                 \omega\gamma\alpha+\rho & \omega\gamma \\
                 -\omega\gamma\alpha\beta & -\omega\gamma\beta+\rho \\
               \end{array}
             \right).
\label{add2}\en
For any such operator, using some general facts arising from  the general pseudo-fermionic structure, we have proposed a simple method, based on ladder operators, to deduce the eigenvalues and eigenvectors of $H$ and of $H^\dagger$. Intertwining relations can also be deduced, and different scalar products on the Hilbert space $\Bbb C^2$ are shown to play a role, in some cases.

Now, $\Theta_2$ can be written as in (\ref{add2}) taking, for instance, $\alpha=-2-\sqrt{3}$, $\beta=\frac{\sqrt{3}+1}{3\sqrt{3}-7}$, $\rho=E_1$, $\alpha_{12}=\sqrt{\frac{1}{8}(38-21\sqrt{3})}=-\beta_{12}$ and $\omega\gamma=\frac{1}{4}(7-3\sqrt{3})(E_1-E_2)$, with $\gamma$ as above.

Of course, this has useful consequences, as discussed in \cite{baginbagbook,baggarg}, since the general framework proposed there for the deformed anti commutation relations can now be used for $\Theta_2$. Viceversa, of course, one could also look now for explicit connections of $\Theta_2$ with some quantum mechanical system, for instance in the context of PT-quantum mechanics.

Finally, since $\Theta_2$ arises out of the 3-by-3 matrix $\Theta_1$, one can imagine that investigating the inverse map $\Theta_2\rightarrow\Theta_1$, the pseudo-fermionic framework (which is originally defined in a two-dimensional space) can be somehow extended to a 3-dimensional Hilbert space, similarly to what was done in \cite{BAG2013}.

\subsection{A first example with $\dim(\Hil)=\infty$}\label{sectex4}

Let $\{\epsilon_n\}$, with $0=\epsilon_0<\epsilon_1<\epsilon_2<\cdots$, and $\{\theta_n\}$ be two sequences of real numbers, and let us assume that, at least for some $n$, $\theta_n$ is not an integer multiple of $\pi$. In this way, $\Im\{e^{i\theta_n}\}\neq0$. Now, let $\F_e=\{e_n, \,n\geq0\}$ be an orthonormal (o.n.) basis for $\Hil$, and let us define the following operator:
$$
D(\Theta_1)=\left\{f\in\Hil:\, \sum_{n=0}^\infty \epsilon_n e^{i\theta_n}\left<e_n,f\right>e_n\in\Hil\right\},
$$
and
$$
\Theta_1 f=\sum_{n=0}^\infty \epsilon_n e^{i\theta_n}\left<e_n,f\right>e_n,
$$
for all $f\in D(\Theta_1)$. Of course, $D(\Theta_1)$ is dense since it contains the linear span of the $e_n$'s, $\Lc_e=l.s.\{e_n\}$. In particular, $\Theta_1 e_n=\epsilon_n e^{i\theta_n} e_n$, for all $n$. This shows that $\epsilon_n e^{i\theta_n}$ are the eigenvalues of $\Theta_1$, not all reals, that $\Theta_1\neq\Theta_1^\dagger$, and that $\varphi_n^{(1)}=e_n$. This is an explicit example showing that the eigenvectors of a manifestly non self-adjoint operator can still form an o.n. basis. Notice also that we are here considering a sort of {\em inverse problem}: rather than deducing the eigensystem out of a given operator, we use a given set of numbers and vectors to define an operator which has this particular set of numbers and vectors as eigensystem.

Now, an operator $X$ having the properties required in Section \ref{sectcase3} can be defined as follows:
$$
D(X)=\left\{f\in\Hil:\, \sum_{n=0}^\infty \sqrt{\epsilon_{n+1}} \left<e_n,f\right>e_{n+1}\in\Hil\right\},
$$
which also contains $\Lc_e$, and
$$
X f=\sum_{n=0}^\infty \sqrt{\epsilon_{n+1}} \left<e_n,f\right>e_{n+1},
$$
for all $f\in D(X)$. The adjoint $X^\dagger$ of $X$ is $X^\dagger g=\sum_{n=0}^\infty \sqrt{\epsilon_{n+1}} \left<e_{n+1},g\right>e_{n}$, for all $g\in D(X^\dagger)$, which again contains $\Lc_e$. $X^\dagger$ behaves as a lowering operator for $\F_e$. Indeed we have
$$
X^\dagger e_n=\left\{
    \begin{array}{ll}
    0, \qquad\qquad \mbox{ if } n=0\\
    \sqrt{\epsilon_n} e_{n-1} \quad\,\, \mbox{ if } n\geq1.\\
    \end{array}
        \right.
$$
Of course, $X$ is a raising operator: $X e_n=\sqrt{\epsilon_{n+1}}\,e_{n+1}$. As for $N_1$ and $N_2$, we have $D(N_j)\supseteq\Lc_e$, $j=1,2$, and we find, for all $f\in\Lc_e$,
$$
N_1f=\sum_{n=0}^\infty\epsilon_n \left<e_n,f\right>e_n, \qquad N_2f=\sum_{n=0}^\infty\epsilon_{n+1} \left<e_n,f\right>e_n.
$$
Then, $N_2^{-1}$ does exist, and a direct computation shows that $\Theta_1N_1f=N_1\Theta_1f$, for all $f\in \Lc_e$, which is left stable by the action of all the operators introduced so far. Hence, $\Lc_e$ plays the role of the set $\D$  introduced in Section \ref{sectcase2}, and the assumptions of Section \ref{sectcase3} are satisfied (in their extended version adapted to infinite dimensional Hilbert spaces).

The set $\F_\psi^{(1)}$ biorthogonal to $\F_\varphi^{(1)}=\F_e$ is, of course, $\F_e$ itself, and equation (\ref{26}) can be explicitly verified. Analogously, all the intertwining equations introduced in Section \ref{sect2} can be explicitly checked. As for $\Theta_2$, we find that
$$
\Theta_2 f=\sum_{n=1}^\infty \epsilon_n e^{i\theta_n}\left<e_{n-1},f\right>e_{n-1},
$$
for all $f\in D(\Theta_2)$, which is the set of vectors for which the series $\sum_{n=1}^\infty \epsilon_n e^{i\theta_n}\left<e_{n-1},f\right>e_{n-1}$ converges in $\Hil$. Once more, $\Lc_e\subseteq D(\Theta_2)$. Now, since
$$
\varphi_n^{(2)}=X^\dagger \varphi_n^{(1)} =X^\dagger e_n =\left\{
    \begin{array}{ll}
    0, \qquad\qquad \mbox{ if } n=0\\
    \sqrt{\epsilon_n} e_{n-1} \quad\,\, \mbox{ if } n\geq1,\\
    \end{array}
        \right.
$$
we can check that $\Theta_2\varphi_n^{(2)}=\epsilon_n\varphi_n^{(2)}$, for all $n\geq0$. Moreover, since $\psi_n^{(1)}=\varphi_n^{(1)}=e_n$, and since $\psi_n^{(2)}=X^\dagger\psi_n^{(1)}$, it follows that $\psi_n^{(2)}=\varphi_n^{(2)}$. We find that
$
\Theta_2^\dagger\psi_n^{(2)}=\overline{\epsilon_n}\psi_n^{(2)},
$
as well as $N_j\varphi_n^{(j)}=\epsilon_n\varphi_n^{(j)}$ and $N_j\psi_n^{(j)}=\epsilon_n\psi_n^{(j)}$, $j=1,2$. Finally, while $\left<\varphi_n^{(1)},\psi_m^{(1)}\right>=\delta_{n,m}$, we have $\left<\varphi_n^{(2)},\psi_m^{(2)}\right>=\epsilon_n\delta_{n,m}$. Hence the general structure of Section \ref{sectcase3} is fully recovered. It might be interesting to notice that the orthogonality of the sets $\F_\varphi^{(j)}$ and $\F_\psi^{(j)}$, $j=1,2$, was expected because their vectors are eigenstates of $N_1$ and $N_2$ with eigenvalues which are not degenerate.

\subsection{A second example with $\dim(\Hil)=\infty$}\label{sectex3}

This example adapts what was originally discussed in \cite{bagintop3} to the present situation. Let $\Theta_1$ be the following (infinite) matrix, acting on the Hilbert space $\Hil=l^2(\Bbb N)$:
$$
\Theta_1=\left(
           \begin{array}{cccccccc}
             \alpha_1 & \beta_1 & 0 & 0 & 0 & 0 & . & . \\
             \beta_1 & \alpha_1 & 0 & 0 & 0 & 0 & . & . \\
             0 & 0 & \alpha_2 & \beta_2 & 0 & 0 & . & . \\
             0 & 0 & \beta_2 & \alpha_2 & 0 & 0 & . & . \\
             0 & 0 & 0 & 0 & \alpha_3 & \beta_3 & . & . \\
             0 & 0 & 0 & 0 & \beta_3 & \alpha_3 & . & . \\
             . & . & . & . & . & . & . & . \\
             . & . & . & . & . & . & . & . \\
           \end{array}
         \right),
$$
where $\alpha_j$ and $\beta_j$ are, in general, complex numbers. Its eigevectors are
$$
\varphi_1^{(1)}=\frac{1}{\sqrt{2}}\left(
  \begin{array}{c}
    1 \\
    -1 \\
    0 \\
    0 \\
    0 \\
    0 \\
    . \\
    . \\
  \end{array}
\right), \, \varphi_2^{(1)}=\frac{1}{\sqrt{2}}\left(
  \begin{array}{c}
    1 \\
    1 \\
    0 \\
    0 \\
    0 \\
    0 \\
    . \\
    . \\
  \end{array}
\right), \, \varphi_3^{(1)}=\frac{1}{\sqrt{2}}\left(
  \begin{array}{c}
    0 \\
    0 \\
    1 \\
    -1 \\
    0 \\
    0 \\
    . \\
    . \\
  \end{array}
\right), \, \varphi_4^{(1)}=\frac{1}{\sqrt{2}}\left(
  \begin{array}{c}
    0 \\
    0 \\
    1 \\
    1 \\
    0 \\
    0 \\
    . \\
    . \\
  \end{array}
\right),
$$
and so on. The corresponding eigenvalues are $\epsilon_1=\alpha_1-\beta_1$, $\epsilon_2=\alpha_1+\beta_1$, $\epsilon_3=\alpha_2-\beta_2$, $\epsilon_4=\alpha_2+\beta_2$, $\ldots$: $\Theta_1\varphi_n^{(1)}=\epsilon_n\varphi_n^{(1)}$. It is clear that, fixing properly $\alpha_j$ and $\beta_j$, all the eigenvalues turn out to be different. Also, if the imaginary parts of some $\alpha_j$ or $\beta_j$ is non zero, then not all the eigenvalues of $\Theta_1$ are real.
It is also important to notice that $\Theta_1$ is densely defined, since it is well defined on the linear span of the $\varphi_n^{(1)}$'s, which form an o.n. basis for $\Hil$.

\vspace{2mm}

{\bf Remarks:--} (1) If each $\beta_j$ is purely imaginary and each $\alpha_j$ is real, then it is clear that $\epsilon_{2n-1}=\overline{\epsilon_{2n}}$, for all $n\geq1$. Then we can think of $\Theta_1$ as a non self-adjoint Hamiltonian in a broken phases, in which the eigenvalues (which are real in the unbroken region) become conjugate in pairs. With this in mind, this example could be relevant in the context of PT quantum mechanics.

(2) In view of the Remark in Section \ref{sectex2} it is interesting to notice that we can introduce a pseudo-fermionic structure even here, defining a pair of operators $a$ and $b$ satisfying $\{a,b\}=\1$ and $a^2=b^2=0$, for each submatrix
$$
[\Theta_1]_j=\left(
               \begin{array}{cc}
                 \alpha_j & \beta_j \\
                 \beta_j & \alpha_j \\
               \end{array}
             \right),
$$
which can be rewritten as in (\ref{add2}). For this, it is sufficient to take $\omega=2\beta_j$, $\alpha=-\beta=1$, $\rho=\alpha_j-\beta_j$ and $\alpha_{12}=-\beta_{12}=\frac{1}{2}$.

\vspace{2mm}

Let now introduce a bounded operator $X$, see \cite{bagintop3}:
$$
\Hil\ni f \rightarrow Xf=\{(Xf)_j, j\in\Bbb N\}=\{\left<\eta_j,f\right>, j\in\Bbb N\},
$$
where $\F_\eta=\{\eta_j\}$ is a tight frame of $\Hil$ defined as follows: $\eta_{2n-1}=\eta_{2n}=\frac{1}{\sqrt{2}}e_n$, $n\geq1$. Here $e_n$ is the $n-th$ vector of the o.n. canonical basis of $l^2(\Bbb N)$, the one with all zero entries except the $n-th$ one, which is equal to one. Then we have $N_2=X^\dagger X=\1$, which is of course bounded and invertible, while
$$
N_1=XX^\dagger=\frac{1}{2}\left(
           \begin{array}{cccccccc}
             1 & 1 & 0 & 0 & 0 & 0 & . & . \\
             1 & 1 & 0 & 0 & 0 & 0 & . & . \\
             0 & 0 & 1 & 1 & 0 & 0 & . & . \\
             0 & 0 & 1 & 1 & 0 & 0 & . & . \\
             0 & 0 & 0 & 0 & 1 & 1 & . & . \\
             0 & 0 & 0 & 0 & 1 & 1 & . & . \\
             . & . & . & . & . & . & . & . \\
             . & . & . & . & . & . & . & . \\
           \end{array}
         \right)=\frac{1}{2}\left(\1+P\right),
$$
Here $\1$ is the identity operator on $\Hil$, and $P$ is a permutation operator acting as follows:
$$
Pc=P\left(
   \begin{array}{c}
     c_1 \\
     c_2 \\
     c_3 \\
     c_4 \\
     c_5 \\
     c_6 \\
     . \\
     . \\
   \end{array}
 \right)=\left(
   \begin{array}{c}
     c_2 \\
     c_1 \\
     c_4 \\
     c_3 \\
     c_6 \\
     c_5 \\
     . \\
     . \\
   \end{array}
 \right),
$$
for all $c\in\Hil$. $N_1$ is also bounded.
We see that $[N_1,\Theta_1]=0$, so that we are in the conditions of Section \ref{sectcase3}. A simple computation shows that $\Theta_2=N_2^{-1}\left(X^\dagger \Theta_1 X\right)=X^\dagger \Theta_1 X$ is a diagonal matrix, with elements $\alpha_k+\beta_k$. In bra-ket notation: $\Theta_2=\sum_{k=1}^\infty (\alpha_k+\beta_k)|e_k\left>\right<e_k| =\sum_{k=1}^\infty \epsilon_k|e_k\left>\right<e_k|$. It is evident that all the {\em original odd eigenvalues} $\epsilon_{2n-1}$ disappear from the game! The reason is simple: each $\varphi_{2n-1}^{(1)}$ belongs to the kernel of $X^\dagger$. Hence $\varphi_{2n-1}^{(2)}=0$. On the other hand, $\varphi_{2n}^{(2)}=X^\dagger \varphi_{2n}^{(1)}=e_n$, for all $n$. Then the set of the eigenvalues of $\Theta_2$ is properly contained in the set of the eigenvalues of $\Theta_1$, but the set of eigenvectors is still a basis for $\Hil$.

\vspace{2mm}

{\bf Remark:--} If $\beta_j$ and $\alpha_j$ are real, and if $\beta_j>\alpha_j>0$, for all $j$, then it is clear that each $\epsilon_{2n}>0$, while each $\epsilon_{2n-1}<0$, for all $n\geq1$. Then, going from $\Theta_1$ to $\Theta_2$ can be seen as a sort of {\em filter}, which removes all the negative eigenvalues present in the original operator $\Theta_1$.

\vspace{2mm}

It is a simple exercise to check that the intertwining relations $X\Theta_2=\Theta_1X$ and $X^\dagger \Theta_1=\Theta_2X^\dagger$ hold true.

Finding now the set $\F_\psi^{(1)}$ is quite easy, due to the fact that $\F_\varphi^{(1)}$ is already an o.n. set. Hence $\psi_n^{(1)}=\varphi_n^{(1)}$, for all $n$. It is also easy to check that $\Theta_1^\dagger\psi_n^{(1)}=\overline{\epsilon_n}\psi_n^{(1)}$. As for $\Theta_2^\dagger$ we get $\Theta_2^\dagger=\sum_{k=1}^\infty (\overline{\alpha_k}+\overline{\beta_k})|e_k\left>\right<e_k|$, and $\psi_n^{(2)}=\varphi_n^{(2)}$. We will briefly return to this example in Section \ref{sectCS}, in connection with bicoherent states.

\subsubsection{A remark, more than an example: nonlinear $\D$-pseudo bosons}\label{sectex41}

In some recent paper the notion of nonlinear $\D$-pseudo bosons has been introduced and analyzed in some details, \cite{bagpnlpb,bagzno1,bagzno2}. An output of this notion is that the eigenvalues and the eigenvectors of the factorized operators $M=ba$ and $M^\dagger=a^\dagger b^\dagger$ can be deduced using some minimal and natural assumptions on $a$ and $b$. Here $a$ and $b$, and their adjoints, are suitable raising and lowering operators on different sets of vectors. More in details, let us consider a strictly increasing sequence $\{\epsilon_n\}$: $0=\epsilon_0<\epsilon_1<\cdots<\epsilon_n<\cdots$. Then, given two
operators $a$ and $b$ on $\Hil$, and a set $\D\subset\Hil$ which is dense in $\Hil$, and which is stable under the action of $a, b, a^\dagger$ and $b^\dagger$,

\begin{defn}
We will say that the triple $(a,b,\{\epsilon_n\})$ is a family of $\D$-non linear pseudo-bosons ($\D$-NLPBs) if the following properties hold:
\begin{itemize}

\item {\bf p1.} a non zero vector $\Phi_0$ exists in $\D$ such that $a\,\Phi_0=0$;

\item {\bf  p2.} a non zero vector $\eta_0$ exists in $\D$ such that $b^\dagger\,\eta_0=0$;

\item {\bf { p3}.} Calling
\be \Phi_n:=\frac{1}{\sqrt{\epsilon_n!}}\,b^n\,\Phi_0,\qquad \eta_n:=\frac{1}{\sqrt{\epsilon_n!}}\,{a^\dagger}^n\,\eta_0, \label{55} \en we
have, for all $n\geq0$, \be a\,\Phi_n=\sqrt{\epsilon_n}\,\Phi_{n-1},\qquad b^\dagger\eta_n=\sqrt{\epsilon_n}\,\eta_{n-1}. \label{56}\en
\item {\bf { p4}.} The set $\F_\Phi=\{\Phi_n,\,n\geq0\}$ is a basis for $\Hil$.

\end{itemize}

\end{defn}

Of course, since $\D$ is stable under the action of $b$ and $a^\dagger$, it follows that $\Phi_n, \eta_n\in \D$, for all $n\geq0$. The set $\F_\eta=\{\eta_n,\,n\geq0\}$ is a basis for $\Hil$ as well. This follows from the fact that
$M\Phi_n=\epsilon_n\Phi_n$ and $M^\dagger\eta_n=\epsilon_n\eta_n$. Therefore, choosing the normalization of $\eta_0$ and $\Phi_0$ in such a way
$\left<\eta_0,\Phi_0\right>=1$, $\F_\eta$ is biorthogonal to the basis $\F_\Phi$. Then, it is possible to check that $\F_\eta$ is the unique
basis which is biorthogonal to $\F_\Phi$.

\vspace{2mm}

We are apparently in the situation considered in Section \ref{sectcase3}, with $\Theta_1=M=ba$, $\varphi_n^{(1)}=\Phi_n$, $X=b$ and $\Theta_2=ab$. In fact, with these choices we have $\Theta_1 X=X\Theta_2$, which is what is required first in Proposition \ref{prop22}. Moreover, the eigenvectors $\psi_n^{(1)}$ must be identified with the $\eta_n$'s\cite{footnote4}. However, Proposition \ref{prop22} cannot be applied here. In fact: first of all $b$ does not admit inverse, so that point $(i)$ of the cited Proposition does not apply. Moreover, there is no reason for $bb^\dagger$ to commute with $\Theta_1$, and for this reason also point $(ii)$ cannot be used in its most relevant part (the construction of the eigenvectors of $\Theta_2$).

However, for completeness,  we will now briefly see that we can still use a different approach, deducing in this way the eigenvectors of many relevant operators involved in the game, but
unfortunately nothing particularly interesting arise. In fact,
 $$\Theta_2\left(a\varphi_n^{(1)}\right)=a\,\Theta_1\varphi_n^{(1)}=\epsilon_n\left(a\varphi_n^{(1)}\right).
 $$
Then $\varphi_n^{(2)}=a\varphi_n^{(1)}$ which is, because of {\bf { p3}.} above, equal to $\sqrt{\epsilon_n}\varphi_{n-1}^{(1)}$, for all $n\geq0$. Here we are simply putting $\varphi_{-1}^{(1)}=0$. Similarly we see that $b^\dagger\psi_n^{(1)}$ are the eigenstates of $\Theta_2^\dagger$. But $b^\dagger\psi_n^{(1)}$ coincides with $\sqrt{\epsilon_n}\psi_{n-1}^{(1)}$, again with the agreement that $\psi_{-1}^{(1)}=0$. In conclusion, in this case, we can still easily find the eigenstates and the eigenvalues of $\Theta_j$ and $\Theta_j^\dagger$, but nothing particularly new appears going from $(\Theta_1,\Theta_1^\dagger)$ to  $(\Theta_2,\Theta_2^\dagger)$: the eigenvectors are the same, but they are just {\em shifted} in their quantum number. This is not really particularly unexpected, since it is a frequent feature in supersymmetric quantum mechanics, where factorizable Hamiltonians play a relevant role, \cite{susyqm}.

\section{Bicoherent states}\label{sectCS}

We will now discuss how coherent states can be defined within the general framework considered in Section \ref{sect2}, and which kind of properties do they have. The working assumptions in this section are the following: (1) $\F_\varphi^{(1)}$ and $\F_\psi^{(1)}$ are biorthogonal bases for $\Hil$, which will be considered in this section to be infinite-dimensional; (2) the eigenvalues $\epsilon_n$ of $\Theta_1$ are positive. In particular, without loss of generality, we can always assume that $0=\epsilon_0<\epsilon_1<\epsilon_2<\cdots$, as it is usually considered in the literature on (generalized) coherent states, see \cite{gazbook} and references therein. Notice that this condition was already assumed in this paper, in Sections \ref{sectex4} and \ref{sectex41}.

Under our assumptions the sets $\Lc_\varphi^{(1)}=l.s.\{\varphi_n^{(1)}\}$ and $\Lc_\psi^{(1)}=l.s.\{\psi_n^{(1)}\}$ are dense in $\Hil$. We define two operators $A_1$ and $B_1^\dagger$ as follows:
$$
\Lc_\varphi^{(1)}\ni f= \sum_{k=0}^N c_k \varphi_k^{(1)}, \quad \Rightarrow \quad A_1 f= \sum_{k=1}^N c_k \sqrt{\epsilon_k} \varphi_{k-1}^{(1)},
$$
and
$$
\Lc_\psi^{(1)}\ni g= \sum_{k=0}^M d_k \psi_k^{(1)}, \quad \Rightarrow \quad B_1^\dagger g= \sum_{k=1}^M d_k \sqrt{\epsilon_k} \psi_{k-1}^{(1)},
$$
for  $M,N<\infty$. Then, in particular we have
\be
A_1\varphi_k^{(1)}=\left\{
    \begin{array}{ll}
    0, \qquad\qquad \mbox{ if } k=0\\
    \sqrt{\epsilon_k} \varphi_{k-1}^{(1)} \quad\,\, \mbox{ if } k\geq1,\\
    \end{array}
        \right.
        \quad \mbox{and}\quad
        B_1^\dagger\psi_k^{(1)}=\left\{
    \begin{array}{ll}
    0, \qquad\qquad \mbox{ if } k=0\\
    \sqrt{\epsilon_k} \psi_{k-1}^{(1)} \quad\,\, \mbox{ if } k\geq1.\\
    \end{array}
        \right.
\label{41}\en
From these formulas, using the biorthogonality and the completeness of $\F_\varphi^{(1)}$ and $\F_\psi^{(1)}$, we also deduce that
\be
A_1^\dagger\psi_k^{(1)}=\sqrt{\epsilon_{k+1}}\,\psi_{k+1}^{(1)}, \qquad B_1\varphi_k^{(1)}=\sqrt{\epsilon_{k+1}}\,\varphi_{k+1}^{(1)},
\label{42}\en
for all $k\geq0$. Hence $A_1$ and $B_1^\dagger$ act as lowering operators, while $A_1^\dagger$ and $B_1$ behave as raising operators, on different sets\cite{footnote5}. A consequence of these equations is that both $\Theta_1$ and $\Theta_1^\dagger$ can be factorized in the following way:
\be
\Theta_1 \varphi_n^{(1)} = B_1A_1 \varphi_n^{(1)}=\epsilon_n \varphi_n^{(1)}, \qquad \Theta_1^\dagger \psi_n^{(1)} = A_1^\dagger B_1^\dagger \psi_n^{(1)}=\epsilon_n \psi_n^{(1)},
\label{43}\en
for all $n\geq0$.

The following  Proposition shows how  bicoherent states can be introduced, and which are their properties.

\begin{prop}\label{thm1} Let us assume that there exist four constants $r_\varphi, r_\psi >0$, and $0\leq \alpha_\varphi, \alpha_\psi\leq \frac{1}{2}$, such that $\|\varphi_n^{(1)}\|\leq r_\varphi^n (\epsilon_n!)^{\alpha_\varphi} $ and $\|\psi_n^{(1)}\|\leq r_\psi^n (\epsilon_n!)^{\alpha_\psi} $, for all $n\geq0$. Let us define
$$
\rho_\varphi = \frac{1}{r_\varphi}\lim_k (\epsilon_{k+1})^{1/2-\alpha_\varphi}, \qquad \rho_\psi = \frac{1}{r_\psi}\lim_k (\epsilon_{k+1})^{1/2-\alpha_\psi}, \qquad  \hat \rho= \lim_k \epsilon_{k+1},
$$
and $\rho:=\min\left(\rho_\varphi, \rho_\psi, \sqrt{\hat\rho}\right)$. Let $C_\rho(0)$ be the circle in the complex plane centered in the origin and with radius $\rho$. Then, defining
\be
N(|z|)=\left(\sum_{k=0}^\infty\frac{|z|^{2k}}{\epsilon_k!}\right)^{-1/2},
\label{44}\en
and
\be
\varphi_1(z)=N(|z|) \sum_{k=0}^\infty\frac{z^{k}}{\sqrt{\epsilon_k!}}\varphi_k^{(1)}, \qquad \psi_1(z)=N(|z|) \sum_{k=0}^\infty\frac{z^{k}}{\sqrt{\epsilon_k!}}\psi_k^{(1)},
\label{45}\en
these are all well defined for $z\in C_\rho(0)$. Moreover, for all such $z$, $\left<\varphi_1(z),\psi_1(z)\right>=1$, $A_1\varphi_1(z)=z\,\varphi_1(z)$ and $B_1^\dagger\psi_1(z)=z\,\psi_1(z)$. Also, if a measure $d\lambda(r)$ exists such that
$\int_0^\rho d\lambda(r) r^{2k}=\frac{\epsilon_k!}{2\pi}$, for all $k\geq0$, then, calling $d\nu(z,\overline{z})=d\lambda(r)\, d\theta$, we have
\be
\int_{C_\rho(0)} d\nu(z,\overline{z}) N(|z|)^{-2}\left<f,\varphi_1(z)\right>\left<\psi_1(z),g\right>=\left<f,g\right>,
\label{46}\en
for all $f,g\in\Hil$.

\end{prop}

\begin{proof}
First of all it is clear that the series in (\ref{44}) converges in $C_\rho(0)$, since it converges if $|z|^2<\hat\rho$. Now, let us compute $\|\varphi_1(z)\|^2$. Using our assumption on $\|\varphi_n^{(1)}\|$ we get
$$
\|\varphi_1(z)\|\leq|N(|z|)|\sum_{k=0}^\infty\frac{|z|^k}{\sqrt{\epsilon_k!}}\left\|\varphi_k^{(1)}\right\| \leq |N(|z|)| \sum_{k=0}^\infty \frac{(|z|r_\varphi)^k}{(\epsilon_k!)^{1/2-\alpha_\varphi}}.
$$
The power series in the right-hand side converges for all $z$ with $|z|<\rho_\varphi$. Analogously, we can prove that $\|\psi_1(z)\|$ is bounded from above by a power series (times $|N(|z|)|$), which converges for all $z$ with $|z|<\rho_\psi$. Concluding, if $|z|<\rho$,  all the relevant power series do converge: hence both $\varphi_1(z)$ and $\psi_1(z)$ are well defined. Moreover, it is easy to check that $\left<\varphi_1(z),\psi_1(z)\right>=1$ for all such $z$'s.

To show that $\varphi_1(z)$ is an eigenstate of $A_1$ we use (\ref{41}):
$$
A_1\varphi_1(z)=N(|z|)A_1\left(\varphi_0^{(1)}+\frac{z}{\epsilon_1!}\varphi_1^{(1)}+\frac{z^2}{\epsilon_2!}\varphi_2^{(1)}
+\frac{z^3}{\epsilon_3!}\varphi_3^{(1)}+\cdots\right)=
$$
$$
=z N(|z|)\left(\varphi_0^{(1)}+\frac{z}{\epsilon_1!}\varphi_1^{(1)}+\frac{z^2}{\epsilon_2!}\varphi_2^{(1)}
+\frac{z^3}{\epsilon_3!}\varphi_3^{(1)}+\cdots\right) = z \,\varphi_1(z).
$$
In the same way, using the lowering properties of $B_1^\dagger$ on the vectors $\psi_n^{(1)}$, we deduce that $B_1^\dagger\psi_1(z)=z\,\psi_1(z)$.

Finally, to check the resolution of the identity in (\ref{46}), we observe that, taken $f,g\in\Hil$,
$$
\int_{C_\rho(0)} d\nu(z,\overline{z}) N(|z|)^{-2}\left<f,\varphi_1(z)\right>\left<\psi_1(z),g\right>=$$
$$=\sum_{k,l=0}^\infty \frac{\left<f,\varphi_k^{(1)}\right>\left<\psi_l^{(1)},g\right>}{\sqrt{\epsilon_k!\epsilon_l!}}\int_0^\rho d\lambda(r) r^k r^l \int_0^{2\pi} d\theta e^{ik\theta} e^{-il\theta}=$$
$$=2\pi \sum_{k=0}^\infty \frac{\left<f,\varphi_k^{(1)}\right>\left<\psi_k^{(1)},g\right>}{\epsilon_k!}\int_0^\rho d\lambda(r) r^{2k}=$$
$$=\sum_{k=0}^\infty \left<f,\varphi_k^{(1)}\right>\left<\psi_k^{(1)},g\right>=\left<f,g\right>.
$$
We have used here the property of the measure $d\lambda(r)$, and the fact that $\F_\varphi^{(1)}$ and $\F_\psi^{(1)}$ are biorthogonal bases.

\end{proof}

\vspace{2mm}

{\bf Remarks:--} (1) This Proposition extends significantly a similar result originally given in \cite{abg}, where only the existence of states similar to our $\varphi_1(z)$ and $\psi_1(z)$ was discussed.

(2) It is clear that we also have
$$
\int_{C_\rho(0)} d\nu(z,\overline{z}) N(|z|)^{-2}|\psi_1(z)\left>\right<\varphi_1(z)|=\1,
$$
which is analogous to (\ref{46}) but expressed in terms of Dirac bras and kets.

(3) The existence of $\rho_\varphi$, $\rho_\psi$ and $\hat\rho$ is not guaranteed a priori. But it is clear that do exist in some particular cases. For instance, when our framework collapses to the one of {\em standard} coherent states, \cite{gazbook}, i.e. when $B_1=A_1^\dagger$, $\epsilon_k=k$ and when $\F_\varphi^{(1)}$ coincides with $\F_\psi^{(1)}$ and is an orthonormal basis. Then we can take $\alpha_\varphi=\alpha_\psi=0$ and $r_\varphi=r_\psi=1$, so that $\rho_\varphi=\rho_\psi=\hat\rho=\infty$. Hence we have convergence in all the complex plane. Another example is discussed in Section \ref{sectanexample}. Other examples can be found in \cite{abg}.

\vspace{2mm}

So far we have considered bicoherent states  constructed  using the eigenstates of $\Theta_1$ and of $\Theta_1^\dagger$. In Section \ref{sect2} we have shown how, to these operators and these eigenvectors, we can associate new operators ($\Theta_2$ and $\Theta_2^\dagger$) and new eigenvectors, $\varphi_k^{(2)}$ and $\psi_k^{(2)}$. It is natural, therefore, to check if also these vectors can be used to define a new pair of bicoherent states, and how.

Before starting, we remind that while $\left<\varphi_n^{(1)},\psi_k^{(1)}\right>=\delta_{n,k}$, $\left<\varphi_k^{(2)},\psi_n^{(2)}\right>= \tilde k_n\delta_{k,n}$, see (\ref{28bis}). This fact has consequences in the definition of $\varphi_2(z)$ and $\psi_2(z)$, which we now introduce as follows
\be
\varphi_2(z)=N(|z|) \sum_{k=0}^\infty\frac{z^{k}}{\sqrt{\epsilon_k!\,\tilde k_k}}\varphi_k^{(2)}, \qquad \psi_2(z)=N(|z|) \sum_{k=0}^\infty\frac{z^{k}}{\sqrt{\epsilon_k!\,\tilde k_k}}\psi_k^{(2)}.
\label{47}\en
We are assuming, for the moment, that all the $\tilde k_n$'s are different from zero. Let us now observe that $\|\varphi_n^{(2)}\|^2=\|X^\dagger\varphi_n^{(1)}\|^2=\tilde k_n\|\varphi_n^{(1)}\|^2$ and $\|\psi_n^{(2)}\|^2=\|X^\dagger\psi_n^{(1)}\|^2=\tilde k_n\|\psi_n^{(1)}\|^2$. Then, under the assumptions of Proposition \ref{thm1}, with these definitions, both $\varphi_2(z)$ and $\psi_2(z)$ are well defined for all $z\in C_\rho(0)$. Moreover, once again, $\left<\varphi_2(z),\psi_2(z)\right>=1$. If the measure $d\lambda(r)$ is now replaced by a new measure, $d\tilde\lambda(r)$, satisfying the new moment problem
$\int_0^\rho d\tilde\lambda(r) r^{2k}=\frac{\epsilon_k!\tilde k_k}{2\pi}$, for all $k\geq0$, then
\be
\int_{C_\rho(0)} d\tilde\nu(z,\overline{z}) N(|z|)^{-2}|\varphi_2(z)\left>\right<\psi_2(z)|=\sum_{k=0}^\infty |\varphi_k^{(2)}\left>\right<\psi_k^{(2)}|,
\label{48}\en
where $d\tilde\nu(z,\overline{z})=d\tilde\lambda(r) d\theta$. Notice that, in general, this is not necessarily equal to the identity operator, except if $\F_\varphi^{(2)}$ and $\F_\psi^{(2)}$ are also bases, which is not granted a priori.  It might happen, for instance, that $\F_\varphi^{(2)}$ and $\F_\psi^{(2)}$ are not bases, but they are complete in $\Hil$. This is, for instance, what is observed in several applications involving the so-called $\D$-pseudo bosons, \cite{baginbagbook}, where several sets of eigenvectors of non self-adjoint operators turn out to be complete in $\Hil$, but not bases\cite{footnote6} for $\Hil$. In both cases, the linear span of the $\varphi_n^{(2)}$'s, $\Lc_\varphi^{(2)}$, and of the $\psi_n^{(2)}$'s,  $\Lc_\psi^{(2)}$, are again dense in $\Hil$. For this reason we can still introduce two (in general, unbounded) operators on these (dense) sets. Again, the difference of mutual normalization between $(\F_\varphi^{(1)},\F_\psi^{(1)})$ and $(\F_\varphi^{(2)},\F_\psi^{(2)})$, suggests to change a little bit the definition originally given for $A_1$ and $B_1^\dagger$. In fact, it is now convenient to define
$$
\Lc_\varphi^{(2)}\ni f= \sum_{k=0}^N c_k \varphi_k^{(2)}, \quad \Rightarrow \quad A_2 f= \sum_{k=1}^N c_k \sqrt{\frac{\epsilon_k\tilde k_k}{\tilde k_{k-1}}}\, \varphi_{k-1}^{(2)},
$$
and
$$
\Lc_\psi^{(2)}\ni g= \sum_{k=0}^M d_k \psi_k^{(2)}, \quad \Rightarrow \quad B_2^\dagger g= \sum_{k=1}^M d_k \sqrt{\frac{\epsilon_k\tilde k_k}{\tilde k_{k-1}}}\, \psi_{k-1}^{(2)},
$$
so that, in particular,
\be
A_2\varphi_k^{(2)}=\left\{
    \begin{array}{ll}
    0, \qquad\qquad\quad\, \mbox{ if } k=0\\
    \sqrt{\frac{\epsilon_k\tilde k_k}{\tilde k_{k-1}}}\,\varphi_{k-1}^{(2)} \quad\,\, \mbox{ if } k\geq1,\\
    \end{array}
        \right.
        \quad\mbox{and}\quad
        B_2^\dagger\psi_k^{(2)}=\left\{
    \begin{array}{ll}
    0, \qquad\qquad\quad \mbox{ if } k=0\\
    \sqrt{\frac{\epsilon_k\tilde k_k}{\tilde k_{k-1}}}\, \psi_{k-1}^{(2)} \quad\,\, \mbox{ if } k\geq1.\\
    \end{array}
        \right.
\label{49}\en
Now, it is obvious that $A_2\varphi_2(z)=z\,\varphi_2(z)$ and $B_2^\dagger\psi_2(z)=z\,\psi_2(z)$, which means that these states are eigenstates of some suitable lowering operators. Summarizing, for $\varphi_2(z)$ and $\psi_2(z)$ we can deduce the same conclusions as in Proposition \ref{thm1}, except that (\ref{46}) must be replaced by equation (\ref{48}).

Using $A_2$, $B_2^\dagger$ and their adjoints it is also possible to see that $\Theta_2$ and $\Theta_2^\dagger$ admit a sort of factorization. In fact, we can see that
\be
A_2^\dagger\psi_k^{(2)}=\sqrt{\frac{\epsilon_{k+1}\tilde k_k}{\tilde k_{k+1}}}\,\psi_{k+1}^{(2)}, \qquad B_2\varphi_k^{(2)}=\sqrt{\frac{\epsilon_{k+1}\tilde k_k}{\tilde k_{k+1}}}\,\varphi_{k+1}^{(2)},
\label{410}\en
for all $k\geq0$. Therefore, in analogy with equation (\ref{43}), we find that
\be
\Theta_2 \varphi_n^{(2)} = B_2A_2 \varphi_n^{(2)}=\epsilon_n \varphi_n^{(2)}, \qquad \Theta_2^\dagger \psi_n^{(2)} = A_2^\dagger B_2^\dagger \psi_n^{(2)}=\epsilon_n \psi_n^{(2)},
\label{411}\en
for all $n\geq0$.

\vspace{2mm}

{\bf Remark:--} As already mentioned, in what we have done so far, it is important the fact that each $\varphi_{n}^{(1)}\notin\ker(X^\dagger)$, so that $\varphi_{n}^{(2)}\neq0$ and $\tilde k_n\neq0$. In fact, if this is not so and if for some index $\hat n$ we have $\tilde k_{\hat n}=0$, then many of the previous formulas have problems, and $\varphi_2(z)$ and $\psi_2(z)$ cannot be defined, apparently. This is not really so: let us call $\K=\{n\in{\Bbb N}_0: \varphi_{n}^{(1)}\in\ker(X^\dagger)\}$, and $\K^c={\Bbb N}_0\setminus\K$. It is useful (but not mandatory) to assume that $0\in\K^c$. Of course, $\K^c$ is a fully ordered set of natural numbers (with zero), so we can relabel the vectors $\varphi_n^{(2)}$'s and the  $\epsilon_n$'s in an unique way: $\tilde\varphi_l^{(2)}=\varphi_{n_l}^{(2)}$, $\tilde\epsilon_l=\epsilon_{n_l}$, for all $l\in\Bbb N_0$. Here $0=n_0<n_1<n_2<n_3<\cdots$.
Then we introduce the set $\F_{\tilde\varphi}^{(2)}=\{\tilde\varphi_l^{(2)}\}$
and, in the same way, a second set $\F_{\tilde\psi}^{(2)}=\{\tilde\psi_l^{(2)}\}$, extracted out of $\F_\psi^{(2)}$. Since $n_k=n_l$ if and only if $k=l$, we deduce that $\F_{\tilde\varphi}^{(2)}$ and $\F_{\tilde\psi}^{(2)}$ are biorthogonal, $\left<\tilde\varphi_l^{(2)},\tilde\psi_n^{(2)}\right>=\tilde k_l\,\delta_{l,n}$, with all the $\tilde k_l$ which are now, by construction, different from zero. Then our new bicoherent states, which replace those in (\ref{47}), are the following:
\be
\tilde\varphi_2(z)= \tilde N(|z|) \sum_{k=0}^\infty\frac{z^{k}}{\sqrt{\tilde\epsilon_k!\,\tilde k_k}}\tilde\varphi_k^{(2)}, \qquad \tilde\psi_2(z)=\tilde N(|z|) \sum_{k=0}^\infty\frac{z^{k}}{\sqrt{\tilde\epsilon_k!\,\tilde k_k}}\tilde\psi_k^{(2)}.
\label{411bis}\en
Nothing really changes with respect with what we have done before. The only difference has to do with the nature of the sets $\F_{\tilde\varphi}^{(2)}$ and $\F_{\tilde\psi}^{(2)}$, which might no longer be complete (or basis),  even when $\F_\varphi^{(2)}$ and $\F_\psi^{(2)}$  are complete (or basis). The existence of these vectors, and their other properties, can now still be deduced without major differences with respect to what we have shown before.

\subsection{Quantization via bicoherent states}

The problem of quantizing a classical system has a very long story, and can be approached in several ways. We refer to \cite{gazbook,aliengl} for many details on this topic. What we want to do here is to show that bicoherent states can also be used for this purpose. In particular we will show that $A_1$ and $B_1$ in (\ref{41}) and (\ref{42}) can be seen as the result of a {\em suitable quantization} of $z$ and $\overline z$. More explicitly, we will now prove that
\be
\left<f,A_1g\right>=\left<f,\left(\int_{C_\rho(0)} d\nu(z,\overline{z}) N(|z|)^{-2}\,z\,|\varphi_1(z)><\psi_1(z)|\right)g\right>,
\label{add3}\en
and
\be
\left<f,B_1g\right>=\left<f,\left(\int_{C_\rho(0)} d\nu(z,\overline{z}) N(|z|)^{-2}\,\overline{z}\,|\varphi_1(z)><\psi_1(z)|\right)g\right>,
\label{add4}\en
for all $f\in\Hil$ and $g\in D(A_1)\cap D(B_1)$, which is dense in $\Hil$ since it contains $\Lc_\varphi^{(1)}$.

To prove (\ref{add3}), we observe first that, because of the property of $d\lambda(r)$ in $d\nu(z,\overline{z})=d\lambda(r)\, d\theta$,
$$
\int_{C_\rho(0)} d\nu(z,\overline{z})z^{k+1}\,{\overline z}^l=\epsilon_{l}!\,\delta_{k+1,l}.
$$
Therefore, after some computations very similar to those needed to prove the resolution of the identity,
\be
\left<f,\left(\int_{C_\rho(0)} d\nu(z,\overline{z}) N(|z|)^{-2}\,z\,|\varphi_1(z)><\psi_1(z)|\right)g\right>=\sum_{k=0}^\infty \sqrt{\epsilon_{k+1}} \left<f,\varphi_k^{(1)}\right>\left<\psi_{k+1}^{(1)},g\right>,
\label{add5}\en
which is equal to $\left<f,A_1g\right>$. In fact, since $\F_\varphi^{(1)}$ and $\F_\psi^{(1)}$ are biorthogonal bases for $\Hil$, we can write $g=\sum_{k=0}^\infty\left<\psi_k^{(1)},g\right>\varphi_k^{(1)}$. Then, recalling that $g\in D(A_1)$,
$$
A_1g=\sum_{k=0}^\infty\left<\psi_k^{(1)},g\right>A_1\varphi_k^{(1)}=\sum_{k=1}^\infty\left<\psi_k^{(1)},g\right>\sqrt{\epsilon_k}
\varphi_{k-1}^{(1)}=\sum_{k=0}^\infty\left<\psi_{k+1}^{(1)},g\right>\sqrt{\epsilon_{k+1}}
\varphi_{k}^{(1)},
$$
which returns the RHS of (\ref{add5}) after taking the scalar product with $f$. Formula (\ref{add4}) can be proved in a similar way.

To get the operators $A_1^\dagger$ and $B_1^\dagger$ we can repeat the same steps, but with the role of the $\varphi$'s and $\psi$'s exchanged. Of course, more operators could be written in terms of $z$, $\overline z$, and functions of these.

\subsection{An example}\label{sectanexample}

What we are going to discuss here is based on the example discussed in Section \ref{sectex3}. In particular, we will adopt the following choice: $\beta_k=\alpha_1$ and $\alpha_k=(4k-3)\alpha_1$, for all $k=1,2,3,\ldots$. We also fix $\alpha_1\in\Bbb R$. Hence the eigenvalues of $\Theta_1$ become $\hat \epsilon_n=\epsilon_{n+1}=2n\alpha_1$, $n\geq0$. The reason for introducing these $\hat\epsilon_n$, with the property that $\hat\epsilon_0=0$, and the related $\hat\varphi_n^{(1)}=\varphi_{n+1}^{(1)}$, $n=0,1,2,3,\ldots$,
 is that this  was required at the beginning of this section, and was used in the derivation of some formulas. Of course, since $\hat\psi_n^{(1)}=\hat\varphi_n^{(1)}$, the bicoherent states $\varphi_1(z)$ and $\psi_1(z)$ coincide:
\be
\varphi_1(z)=\psi_1(z)=N_1(|z|)\sum_{n=0}^\infty \frac{z^n}{\sqrt{\hat\epsilon_n!}}\,\hat\varphi_n^{(1)}.
\label{32}\en
Now, since $\hat\epsilon_n!=(2\alpha_1)^nn!$, we deduce that $N_1(|z|)=e^{-\frac{|z|^2}{4\alpha_1}}$, for all $z\in \Bbb C$. Formula (\ref{32}) produces now
$$
\varphi_1(z)=\psi_1(z)=\frac{1}{\sqrt{2}}e^{-\frac{|z|^2}{4\alpha_1}}\left(
                                                                       \begin{array}{c}
                                                                         1+\frac{z}{\sqrt{2\alpha_1}} \\
                                                                         -1+\frac{z}{\sqrt{2\alpha_1}} \\
                                                                         \frac{z^2}{2\alpha_1\sqrt{2!}}\left(1+\frac{z}{\sqrt{2\alpha_13!}}\right) \\
                                                                         \frac{z^2}{2\alpha_1\sqrt{2!}}\left(-1+\frac{z}{\sqrt{2\alpha_13!}}\right) \\
                                                                         \cdots \\
                                                                         \cdots \\
                                                                         \cdots \\
                                                                         \cdots \\
                                                                       \end{array}
                                                                     \right).
$$
Since $\|\varphi_k^{(1)}\|=\|\psi_k^{(1)}\|=1$ for all $k$, the hypotheses on the norms of the vectors in Proposition \ref{thm1} are satisfied by simply taking $r_\varphi=r_\psi=1$, and $\alpha_\varphi=\alpha_\psi=0$. Then we have $\rho_\varphi=\rho_\psi=\infty$. In other words: $\varphi_1(z)$ and $\psi_1(z)$ coincide and exist for all $z\in \Bbb C$.

As for $\varphi_2(z)$ and $\psi_2(z)$ the situation is a bit more complicated, since,  $\varphi_{2n+1}^{(2)}=\psi_{2n+1}^{(2)}=0$, for all $n\in\Bbb N_0$. Hence we are in the situation discussed in the Remark at the end of Section \ref{sectCS}, with $\K^c=\{2n, \,n\geq0\}$. We find $\tilde\varphi_l^{(2)}=\varphi_{2l}^{(2)}=e_l$ and $\tilde\psi_l^{(2)}=\psi_{2l}^{(2)}=e_l$, so that $\left<\tilde\varphi_l^{(2)},\tilde\psi_n^{(2)}\right>=\delta_{l,n}$. Then $\tilde k_n=1$ for all $n$. Also $\tilde\epsilon_l=\hat\epsilon_{2l}=(2\alpha_1)^{2l}(2l)!$, and
$$
\tilde\varphi_2(z)=\tilde\psi_2(z) = \tilde N(|z|) \sum_{l=0}^\infty\frac{(z/\alpha_1)^l}{\sqrt{(2l)!}}\,e_l,
$$
which is defined in all of $\Bbb C$. The normalization turns out to be $\tilde N(|z|)=\left(\cosh\left(\frac{|z|}{\alpha_1}\right)\right)^{-1/2}$, for all $z\in\Bbb C$. The properties of these states (mutual normalization, resolution of the identity, etc.) easily follow from the general results proved in this section.

\section{Conclusions}

In this paper we have discussed how to use some ideas coming from the theory of intertwining operators and other ideas concerning non self-adjoint Hamiltonians, to construct new exactly solvable models. We have shown that this, in most of the cases, is really different from the standard situation where similarity conditions between different Hamiltonians can be established. We have found that, in all the cases considered here, useful intertwining relations can be deduced. We have also seen that bicoherent states can naturally be introduced, and we have found conditions for these states to exist and to satisfy several properties which are usually required to mostly all classes of coherent states, like the resolution of the identity and the fact of being eigenstates of certain lowering operators.

We believe that much more can still be deduced within the framework proposed here, in particular for what concerns the bicoherent states or for the deduction of more solvable models. These are two aspects which are presently under deeper investigation.

\section*{Acknowledgements}
The author would like to acknowledge  support from the
   Universit\`a di Palermo and from Gnfm.


\begin{thebibliography}{99}
\bibitem{footnote1}What only exists is a left inverse of $X$, the matrix $\frac{2}{3}\,X^\dagger$, but not an inverse, since $X\left(\frac{2}{3}X^\dagger\right)\neq\1$.
\bibitem{footnote2} A detailed analysis of similarity of operators in the context of PT-quantum mechanics can be found in \cite{antinbagbook}.
\bibitem{footnote3} This might appear strange at a first glance. We refer to Section \ref{sectcase2} for a simple example in which $X^{-1}$ does not exist, but $(X^\dagger X)^{-1}$ does. In that example, $X$ is an operator between two different, finite-dimensional, Hilbert spaces, $\Bbb C^3$ and $\Bbb C^2$.
\bibitem{footnote4} Incidentally we observe that another intertwining relation does also exist between $\Theta_1$ and $\Theta_2$. Indeed we have $a\Theta_1=\Theta_2a$.
\end{thebibliography}
\end{document}